\documentclass[runningheads]{llncs}
\makeatletter
\AtBeginDocument{%
  \@ifpackageloaded{hyperref}
  {\def\@doi#1{\href{https://doi.org/#1}
      {\ttfamily https://doi.org/#1}\egroup}}
  {\def\@doi#1{\ttfamily https://doi.org/#1\egroup}}
  \def\doi{\bgroup\catcode`\_=12\relax\@doi}}
\makeatother
\usepackage[utf8]{inputenc}
\usepackage{listings}
\usepackage{hyperref}
\usepackage{graphicx}
\usepackage{xcolor}
\usepackage{mathrsfs}
\usepackage{wrapfig}
\usepackage[T1]{fontenc}
\usepackage{mathpartir}
\usepackage{butterma}
\authorrunning{É. Lubat et al.}
\idline{E. André and M. Stoelinga (Eds.): FORMATS 2019, LNCS 11750}
\doiline{\doi{10.1007/978-3-030-29662-9\_5}}
\renewcommand{\year}{2019}
\pdfpagesattr{/CropBox [92 72 533 778]} 
\usepackage{latexsym}
\usepackage{textcomp}
\usepackage{amsmath,amsfonts,amstext,amssymb,verbatim}
\usepackage{stmaryrd}
\usepackage[nointegrals]{wasysym} 
\usepackage{footnote}
\usepackage{multirow} 
\usepackage{array}
\usepackage{xspace}
\usepackage{units}
\usepackage{enumitem}
\usepackage{subfig}
\usepackage{alltt}
\usepackage{algorithmic}
\usepackage{placeins}
\usepackage{calculator}
\usepackage{siunitx}
\usepackage{tikz}
\usetikzlibrary{arrows}
\usetikzlibrary{decorations.markings}
\usetikzlibrary{shadows}
\tikzset{
  big stealth/.style={
    decoration={markings,mark=at position -(0.1pt) with {\arrow[scale=2*\scale]{stealth}}},
    postaction={decorate},
    shorten >=0.4pt}}
\tikzset{
  big ring/.style={
    decoration={markings,mark=at position -(0.1pt) with {\arrow[scale=1.5*\scale]{o}}},
    postaction={decorate},
    shorten >=8pt*\scale}}
\tikzset{
  big disc/.style={
    decoration={markings,mark=at position -(0.1pt) with {\arrow[scale=1.5*\scale]{*}}},
    postaction={decorate},
    shorten >=8pt*\scale}}
\tikzset{
  big box/.style={
    decoration={markings,mark=at position -(0.1pt) with {\arrow[scale=1.5*\scale]{open square}}},
    postaction={decorate},
    shorten >=8pt*\scale}}
\tikzset{
  big tile/.style={
    decoration={markings,mark=at position -(0.1pt) with {\arrow[scale=1.5*\scale]{square}}},
    postaction={decorate},
    shorten >=8pt*\scale}}
\tikzstyle{place}=[circle, very thick, fill, top color=white, bottom color=white, draw=black, minimum size=40pt, drop shadow]
\tikzstyle{state}=[circle, very thick, fill, top color=white, bottom color=white, draw=black, minimum size=40pt, drop shadow]
\tikzstyle{trans}=[rectangle, very thick, fill, top color=white, bottom color=white, draw=black, minimum size=32pt, drop shadow]
\tikzstyle{arc}=[thick, big stealth, black]
\tikzstyle{read}=[thick, big disc, black]
\tikzstyle{inhibitor}=[thick, big ring, black]
\tikzstyle{stopwatch}=[thick, big tile, black]
\tikzstyle{stopwatchinhibitor}=[thick, big box, black]
\tikzstyle{priority}=[thick, big stealth, orange]
\tikzstyle{enabling}=[thick, big disc, orange]
\tikzstyle{disabling}=[thick, big ring, orange]
\tikzstyle{token}=[circle, fill, draw=black, minimum size=4pt]
\tikzstyle{glob-options}=[label distance=6pt*\scalenodes*\scale,x=1pt,y=-1pt,scale=\scale,every node/.style={transform shape}]
\tikzstyle{virtual}=[circle, draw=white, minimum size=20pt]
\usepackage{pgfplotstable}
\usepackage{booktabs} 
\usepackage{tabularx}
\newcolumntype{Y}{>{\hfill\arraybackslash}X}
\newcolumntype{C}{>{${}}c<{{}$}}
\newcolumntype{L}{>{${}}l<{{}$}}
\newcolumntype{R}{>{${}}r<{{}$}}
\newcommand{\vect}[1]{\ensuremath{\mathbf{#1}}}

\newcommand{\Pre}{\vect{Pre}} 
\newcommand{\Post}{\vect{Post}}

\newcommand{\Lab}{\ensuremath{\mathcal{L}}\xspace}

\newcommand{\TPN}[0]{{TPN}\xspace}
\newcommand{\TTPN}[0]{{PTPN}\xspace}

\newcommand*{\dotleq}{\mathrel{\dot{\leqslant}}}
\newcommand*{\dotgeq}{\mathrel{\dot{\geqslant}}}

\newcommand{\arel}{\mathrel{\mathcal{R}}}

\newcommand{\SG}{{{S}}}
\newcommand{\Nat}{\ensuremath{\mathbb{N}}\xspace}

\newcommand{\pRat}{\ensuremath{\mathbb{Q}_{\ge 0}}\xspace}

\newcommand{\Itrv}{\ensuremath{\mathbb{I}\xspace}}
\newcommand{\proj}[1]{\mathop{\mathtt{\scalebox{0.8}{\raisebox{0.4ex}{\#}}}{_#1}}}
\makeatletter
\newcommand{\dotminus}{\mathbin{\text{\@dotminus}}}
\newcommand{\@dotminus}{%
  \ooalign{\hidewidth\raise1ex\hbox{.}\hidewidth\cr$\m@th-$\cr}%
}
\makeatother

\makeatletter
\def \rightarrowfill{\m@th\mathord{\smash-}\mkern-6mu%
  \cleaders\hbox{$\mkern-2mu\mathord{\smash-}\mkern-2mu$}\hfill
  \mkern-6mu\mathord\rightarrow}
\makeatother

\makeatletter
\def \Rightarrowfill{\m@th\mathord{\smash-}\mkern-6mu%
  \cleaders\hbox{$\mkern-2mu\mathord{\smash-}\mkern-2mu$}\hfill
  \mkern-6mu\mathord\Rightarrow}
\makeatother

\makeatletter
\def \rightarrowfill{\m@th\mathord{\smash-}\mkern-6mu%
  \cleaders\hbox{$\mkern-2mu\mathord{\smash-}\mkern-2mu$}\hfill
  \mkern-6mu\mathord\rightarrow}
\makeatother

\makeatletter
\def \Rightarrowfill{\m@th\mathord{\smash=}\mkern-6mu%
  \cleaders\hbox{$\mkern-2mu\mathord{\smash=}\mkern-2mu$}\hfill
  \mkern-6mu\mathord\Rightarrow}
\makeatother

\makeatletter
\def \midrightarrowfill{\m@th\mathord{\smash{\raisebox{.2ex}{$\scriptscriptstyle\mid$}}\!\!\,-}\mkern-6mu%
  \cleaders\hbox{$\mkern-2mu\mathord{\smash-}\mkern-2mu$}\hfill
  \mkern-6mu\mathord\rightarrow}
\makeatother

\makeatletter
\def \midRightarrowfill{\m@th\mathord{\smash{\raisebox{.1ex}{$\scriptstyle\mid$}}\!\!\!=}\mkern-6mu%
  \cleaders\hbox{$\mkern-2mu\mathord{\smash=}\mkern-2mu$}\hfill
  \mkern-6mu\mathord\Rightarrow}
\makeatother

\newcommand{\overstackrel}[2]{\mathrel{\mathop{#1}\limits^{#2}}}
\newcommand{\trans}[1]{\mathbin{\smash[t]{\overstackrel{\rightarrowfill}{\ #1\ }}}}
\newcommand{\wtrans}[1]{\mathbin{\smash[t]{\overstackrel{\Rightarrowfill}{\ #1\ }}}} 
\newcommand{\dwtrans}[1]{\mathbin{{\overstackrel{\Rightarrowfill}{\ #1\ }}}} 
\newcommand{\ptrans}[1]{\mathbin{\smash[t]{\overstackrel{\midrightarrowfill}{\raisebox{-0.4ex}[0ex][-0.4ex]{$\scriptstyle\ #1\ $}}}}}
\newcommand{\dptrans}[1]{\mathbin{{\overstackrel{\midrightarrowfill}{{\raisebox{-0.4ex}[0ex][-0.4ex]{$\scriptstyle\ #1\ $}}}}}}
\newcommand{\interp}[1]{{[\![} {#1} {]\!]}}

\newcommand{\TT}[0]{\mathscr{T}_1}

\newcommand{\tI}[1][]{\ensuremath{\varphi_{#1}}}

\newcommand{\sI}[1]{\ensuremath{x_{#1}}}

\newcommand{\E}{{\cal E}}
\newcommand{\SIF}{\ensuremath{{\bf I}_s}\xspace}

\def\eqdef{\mathrel{\smash{\stackrel{{\scriptscriptstyle\mathrm{def}}}{=}}}}

\begin{document}
\title{A State Class Construction for Computing the Intersection of
  Time Petri Nets Languages}
\titlerunning{Computing the Language Intersection of TPN}
%
\author{Éric Lubat\inst{1} \and Silvano {Dal
    Zilio}\inst{1} \and\\ Didier Le
  Botlan\inst{1} \and Yannick Pencolé\inst{1} \and Audine
  Subias\inst{1}}
\pagestyle{headings}
%
\institute{LAAS-CNRS, Université de Toulouse, CNRS, INSA,
  Toulouse, France\\ \email{name.surname@laas.fr}}

\maketitle
\thispagestyle{electronic}

%
%
\begin{abstract}
  We propose a new method for computing the language intersection of
  two Time Petri nets (TPN); that is the sequence of labels in timed
  traces common to the execution of two TPN. Our approach is based on
  a new product construction between nets and relies on the State
  Class construction, a widely used method for checking the behaviour
  of TPN. We prove that this new construct does not add additional
  expressive power, and yet that it can leads to very concise
  representation of the result. We have implemented our approach in a
  new tool, called {Twina}. We report on some experimental results
  obtained with this tool and show how to apply our approach on two
  interesting problems: first, to define an equivalent of the
  twin-plant diagnosability methods for TPN; then as a way to check
  timed properties without interfering with a system.  
\end{abstract}

\section{Introduction}
\label{sec:introduction}

Formal languages, and the problem of efficiently checking intersection
between languages, play an important role in formal verification. For
instance, automata-theoretic approaches to model-checking often boils
down to a language emptiness problem; that is finding whether there is
a trace, in a system, that is also ``in the negation of a
property''~\cite{kupferman2000automata}. Similarly, in the study of
Discrete Event Systems~\cite{ramadge1989control}, basic
control-theoretic properties are often expressed in terms of language
properties and language composition. We consider examples of these two
problems at the end of this paper.

In this context, there is a large body of research where systems are
expressed using Petri nets (PN). Indeed, PN are well-suited for
modelling notions such as concurrency or causality in a very compact
way; and they can be used for verification by building a {Labeled
  Transition System} out of them. Just as important, PN come equipped
with a structural construct for \emph{synchronous composition}, that
coincides with language intersection when the set of labels of the
nets are equal. Unfortunately, the situation is not as simple when we
consider extensions of Petri nets that deal with time.

In this paper, we propose a new method for computing the language
intersection of two \emph{Time Petri nets}
(TPN)~\cite{Merlin74astudy,BPV06}.  This problem is quite complex and
is hindered by two main problems. First, the state space associated
with a \TPN is typically infinite when we work with a dense time
model; that is when time delays can be arbitrarily small. Therefore we
need to work with an abstraction of their transition system. Second,
there is no natural way to define the (structural) composition of two
transitions that have non-trivial time constraints (meaning different
from the interval $[0,\infty[$). These problems limit the possibility
for compositional reasoning on \TPN.

A solution to the first problem was proposed by Berthomieu and
Menasche in~\cite{berthomieu1983enumerative}, where they define a
state space abstraction based on \emph{state classes}.  This approach
is used in several model-checking tools, such as
Romeo~\cite{gardey2005romeo} and Tina~\cite{BRV04} for instance.  In
the following, we propose a simple solution to overcome the second
problem. Our approach is based on an extension of \TPN with a
dedicated product operator, called \emph{Product \TPN}, that can be
viewed as an adaptation of Arnold-Nivat synchronization
product~\cite{arnold_nivats_2002} to the case of \TPN. We show that it
is possible to extend the state class construction to this new
extension, which gives an efficient method for computing the
intersection of two \TPN when the nets are bounded.\\

\noindent\textbf{Verification of Time Petri Nets.}
In the following, we consider \TPN where transitions may have
observable labels.  In this context, an \emph{execution} is the
timed-event word obtained by recording the transitions that have been
fired together with the delays between them.  Our goal is to provide a
method for symbolically computing the set of executions that are
{common} to two labeled \TPN.  Without time, it is well-known that we
can compute the language of a net from its marking graph. This gives a
Labeled Transition System (LTS); an automaton that is finite as soon
as the net is bounded. Likewise, we can compute the (language)
intersection of two timed nets by computing the LTS of their
\emph{synchronous composition}, denoted $N_1 \| N_2$
thereafter. Actually, like in the untimed case, we are more interested
by the \emph{synchronous product} of two languages, rather than by
their intersection.

The situation is quite different when we take time into
account. Indeed, we may have fewer traces with a \TPN than with the
corresponding, ``untimed'' net (the one where timing constraints are
deleted). This is because timing constraints may prevent a transition
from firing, but never enable it. One solution to recover a finite
abstraction of the state space is to use the \emph{State Class Graph}
(SCG) construction. Actually, SCG is an umbrella term for a family of
different abstractions, each tailored to a different class of
properties, or to a different extension of \TPN. The first such
construction, called \emph{Linear State Class Graph}
(LSCG)~\cite{berthomieu1983enumerative}, is based on \emph{firing
  domains}, that is the delays before a transition can fire.  The LSCG
preserves the set of reachable markings of a net as well as its
language; which is exactly what is needed in our case. This is also
the construction that we use in Sect.~\ref{sec:2TPN}.

In the following we also mention the \emph{Strong SCG} construction
(SSCG)~\cite{BPV07}, based on \emph{clock domains}, that is the
duration for which a transition has been enabled.  The SSCG preserves
more information than the linear one. For example, we can infer from
clocks when two transitions are enabled ``at the same time'', meaning
we can handle priorities. The added expressiveness of the strong
construction comes at a cost; the SSCG (for a given net) has always
more classes than the corresponding LSCG, sometimes by a very large
amount. (We give some examples of this in
Sect.~\ref{sec:experimental-results}.)  This is why we prefer to use
the LSCG when possible.\\

\noindent\textbf{Related Works and Review of Existing Methods.}
A motivation for our work is that we cannot rely on a synchronous
product of \TPN. Indeed, a major limitation with \TPN is that there
are no sensible way to define the composition of ``non-trivial''
transitions, and therefore no sensible way to define the synchronous
composition of ``non-composable'' \TPN; we say that \emph{a transition
  is trivial} when it is associated to the time interval $[0, \infty[$
and that a net is \emph{composable} when all its observable
transitions are trivial. (We illustrate the problem at the end of
Sect.~\ref{sec:TPN}).  Likewise we cannot rely on the product of their
SCG either. Indeed, the product of two SCG provides an
over-approximation of the expected result, since it cannot trace time
dependencies between events from different nets.

The situation is not the same with other ``timed models''. A notable
example is \emph{Timed Automata}~\cite{alur1994theory}, an extension
of finite automata with variables, also called clocks, whose values
progress synchronously as time elapses. Timed Automata (TA) can use
boolean conditions on clocks to guard transitions and as local
invariants on states. It is also possible to reset a clock when
``firing'' a transition. The classical product operation on finite
automata can be trivially extended to TA: we only need to use the
conjunction of guards, invariants and resets where needed. This
provides a straightforward method for computing the (language)
intersection of two TA, and also a trivial proof that the class of
languages accepted by a TA are closed under intersection. Another
related work is based on the definition of \emph{Timed Regular
  Expressions}~\cite{asarin2002timed}, that provides a timed analogue
of Kleene Theorem for TA.

These results seem to promote Timed Automata as an algebraic model of
choice for reasoning about timed words, and many works have studied
the relation between \TPN and TA. (On another note, we can remark that
even a slight change in semantics may complicate the product
construction; see for instance the case with signal-event
languages~\cite{berard_intersection_2006}.) For instance, Cassez and
Roux~\cite{cassez2006structural} propose a structural encoding of \TPN
into TA that preserves the semantics in the sense of timed
bisimulation, and therefore that preserves timed language
acceptance. This encoding generates one automata and one clock for
every transition in the \TPN and can be extended in order to
accommodate strict timing constraints; that is static time intervals
that have a finite, open bound. Later, Bérard et
al.~\cite{berard2005comparison,berard-TCS-08} showed that \TPN and TA
are indeed equivalent with respect to language acceptance, but that TA
are strictly more expressive in terms of weak timed bisimulation
($\approx$). These results are based on semantic encodings from \TPN
into TA and from TA into \TPN that can be chained together to build an
encoding from a \TPN to an equivalent composable one. A similar result
is also found in~\cite{refIPTPN}, which provides a structural encoding
from a \TPN, $N$, into a composable \TPN that is of size linear with
respect to $N$. But none of these encodings handle timing constraints
that are bounded and right-open.

One of the main difference between TA and \TPN is that, with TA, we
can loose the ability to fire a transition just by waiting long enough
(until some guards become false). The same behaviour can be observed
with \TPN when we add a notion of priorities. In particular,
Berthomieu et al.~\cite{BPV06,BPV07} prove that (bounded) \TPN with
priorities are very close to TA, in the sense of $\approx$. They also
define an extension of \TPN~\cite{refIPTPN} with \emph{inhibitor arcs}
between transitions (similar to priorities) and a dual notion of
\emph{permission arcs}. In this extension, called IPTPN, a net can
always be transformed into a composable one. (We show an example of
this construction in Sect.~\ref{sec:expressiveness}).

All these results can be used to define three different methods for
computing the intersection of \TPN.  {A first method} is to use the
structural translation from \TPN to TA given
in~\cite{cassez2006structural} and then to use the product
construction on TA. This encoding is at the heart of the tool
Romeo~\cite{gardey2005romeo} and has been used to build a TCTL
model-checker for \TPN (which, incidentally, relies on the ``product''
of a net with observers for the formulas). Unfortunately, to the best
of our knowledge, it is not possible to analyse the product of two
nets with Romeo and therefore we have not been able to experiment with
this method.  Moreover, this approach is closer in spirit to the SSCG
construction.

{A second method} is to use the (combination of) encodings defined
in~\cite{berard-TCS-08} to replace a \TPN with an equivalent,
composable one.  Unfortunately, this construction relies on a semantic
encoding that requires the computation of the entire symbolic state
space of the net, and is only applicable on net that have closed
timing constraints; meaning that we cannot use constraints of the form
$[l, h[$ for example. While this method is not usable in practice, it
could be used to prove expressiveness results. For example, it gives a
proof that the set of TPN with closed timing constraints is closed
under intersection; something we silently admitted until now.

{A third method} also relies on generating composable nets as a
preprocessing step. In this case, the idea is to use the IPTPN
of~\cite{refIPTPN}. Like in the first method, the main drawback of
this approach is that we need to use the strong SCG construction,
which means that we could compute much more classes than with a method
based on the LSCG.  We describe the experimental results obtain with
this
method in Sect.~\ref{sec:experimental-results}.\\

\noindent\textbf{Outline of the Paper and Contributions.}
In the next section we define the semantics of \TPN and provide the
technical background necessary for our work.  Section~\ref{sec:2TPN}
contains the semantics of Product \TPN, while our two main results are
given in Sect.~\ref{sec:state-class-abstr}
and~\ref{sec:expressiveness}, where we show that it is possible to
extend the State Class Graph construction to the case of Product \TPN
and that this extension does not add additional expressiveness
power. By construction, our method can be applied even when the \TPN
are not bounded and without any restrictions on the timing
constraints.

We have implemented our approach in a new tool, called
Twina~\cite{twina2019} Before concluding, we report
(Sect.~\ref{section:Twin}) on some experimental results obtained with
this tool. We also show some practical applications for our approach
on two problems: first, to define an equivalent of the twin-plant
diagnosability methods for \TPN; then as a way to check timed
properties without interfering with a system.

\section{Time Petri Nets and other Technical Background}
\label{sec:TPN}

A {\em Time Petri Net} (\TPN) is a net where each transition, $t$, is
decorated with a (static) time interval $ \SIF(t)$ that constrains the
time at which it can fire. A transition is enabled when there are
enough tokens in its input places. Once enabled, transition $t$ can
fire if it stays enabled for a duration $\theta$ that is in the
interval $\SIF(t)$. In this case, $t$ is said \emph{time enabled}.

A \TPN is a tuple $\langle {P},{T},{\Pre},{\Post},m_0,\SIF \rangle$ in
which: $\langle {P},{T},{\Pre},{\Post} \rangle$ is a net (with ${P}$
and ${T}$ the set of places and transitions);
${\Pre},~ {\Post} : {T} \rightarrow {P} \rightarrow \Nat$ are the
precondition and postcondition functions; $m_0 : P \rightarrow \Nat$
is the initial marking; and $\SIF : {T} \rightarrow \Itrv$ is the
\emph{static interval function}. We use $\Itrv$ for the set of all
possible time intervals. To simplify our presentation, we only
consider the case of closed intervals of the form $[l, h]$ or
$[l, +\infty[$, but our results can be extended to the general case.
TPN can be \emph{k-safe}, which means the net has at most $k+1$
reachable markings. We say that a TPN is \emph{safe} when it is
1-safe.

We consider that transitions can be tagged
using a countable set of labels, $\Sigma = \{a, b, \dots\}$. We also
distinguish the special constant $\epsilon$ (not in $\Sigma$) for
internal, silent transitions. In the following, we use a global
labeling function $\Lab$ that associates a unique label in
$\Sigma \cup \{\epsilon\}$ to every transition\footnote{We may assume
  that there is a countable set of all possible transitions
  (identifiers) and that different nets have distinct
  transitions.}. The alphabet of a net is the collection of labels (in
$\Sigma$) associated to its transitions.\\

\noindent\textbf{A Semantics for TPN Based on Firing Domains.}
A \emph{marking} $m$ of a net $\langle {P},{T},{\Pre},{\Post} \rangle$
is defined as a function $m : {P} \rightarrow \Nat$ from places to
natural numbers. A transition $t$ in ${T}$ is {\em enabled} at $m$ if
and only if $m \dotgeq \Pre(t)$ (we use the pointwise comparison between
functions) and ${\cal E}(m)$ denotes the set of transitions enabled at
$m$.

A \emph{state} of a {\TPN} is a pair $s = (m, \tI)$ in which $m$ is a
marking, and $\tI: {T} \to \Itrv$ is a mapping from transitions to
time intervals, also called \emph{firing domains}. Intuitively, if $t$
is enabled at $m$, then $\tI(t)$ contains the dates at which $t$ can
possibly fire in the future. For instance, when $t$ is newly enabled,
it is associated to its static time interval $\tI(t) =
\SIF(t)$. Likewise, a transition $t$ can fire immediately only when
$0$ is in $\tI(t)$ and it cannot remain enabled for more than its
timespan, {\it i.e.} the maximal value in $\tI(t)$.

For a given delay $\theta$ in $\pRat$ and $\iota$ in $\Itrv$, we denote $\iota - \theta$ the
time interval $\iota$ shifted (to the left) by $\theta$:, e.g.
$[l,h] - \theta = [\max(0, l-\theta), \max(0, h-\theta)]$. By
extension, we use $\tI \dotminus \theta$ for the partial function that
associates the transition $t$ to the value $\tI(t) - \theta$. This
operation is useful to model the effect of time passage on the enabled
transitions of a net.

The following definitions are quite standard, see for
instance~\cite{berard2005comparison,BPV06}. The semantics of a \TPN is
a (labeled) Kripke structure $\langle S,S_0,\rightarrow\rangle$ with
only two possible kinds of actions: either $s \trans{a} s'$, meaning
that the transition $t \in T$ is fired from $s$ with $\Lab(t) = a$; or
$s \trans{\theta} s'$, with $\theta \in \pRat$, meaning that time
$\theta$ elapses from $s$. A transition $t$ can fire from the state
$(m,\tI)$ if $t$ is enabled at $m$ and firable instantly. When we fire
a transition $t$ from state $(m, \tI)$, a transition $k$ (with
$k \neq t$) is said to be \emph{persistent} if $k$ is also enabled in
the marking $m - \Pre(t)$, that is if $m - \Pre(t) \dotgeq
\Pre(k)$. The other transitions enabled after firing $t$ are called
\emph{newly enabled}.

\begin{sloppypar}
  \begin{definition}[Semantics]\label{def:tpnstate}
    The semantics of a {\TPN} $N$, with $N$ the net
    $\langle {P}, {T}, {\Pre}, {\Post}, m_0, \SIF \rangle$, is the
    Timed Transition System (TTS)
    $\interp{N} = {\langle \SG, s_0, \rightarrow \rangle}$ where $S$
    is the smallest set containing $s_0$ and closed by $\trans{}$,
    where:
    \begin{itemize}
    \item $s_0 = (m_0, \tI[0])$ is the initial state, with $m_0$ the
      initial marking and $\tI[0](t) = \SIF(t)$ for every $t$ in
      $\E(m_0)$;
    \item the state transition relation
      ${\rightarrow} \subseteq S \times ({{\Sigma} \cup \{\epsilon\} \cup \pRat})\! \times S$
      is the relation such that for all state $(m, \tI)$ in $\SG$:
      \begin{itemize}
      \item[(i)] if $t$ is enabled at $m$, $\Lab(t) = a$ and
        $0 \in \tI(t)$ then $(m, \tI) \trans{a} (m', \tI')$ where
        $m' = m - \Pre(t) + \Post(t)$ and $\tI'$ is a firing function
        such that $\tI'(k) = \tI(k)$ for any persistent transition and
        $\tI'(k) = \SIF(k)$ elsewhere.
      \item[(ii)] if $\theta \dotleq \tI$ 
      \\
      $\forall k \in Enabled(m) , \theta \leq max \tI(k)$ then
        $(m,\tI) \trans{\theta} (m,\tI \dotminus \theta)$.
      \end{itemize}
    \end{itemize}
    \label{def:sg}
  \end{definition}
\end{sloppypar}

Transitions in the case $(i)$ above are called \emph{discrete
  transitions}; those labelled with delays (case $(ii)$) are the
\emph{continuous}, or time elapsing, transitions.
Like with nets, we say that the alphabet of a TTS is the set of
labels, in $\Sigma$, associated to discrete actions. Using labels, we
can define the {product of two TTS} by extending the classical
definition for the product of finite automata.
\begin{definition}[Product of TTS]\label{def:ttsprod}
  Assume $S_1 = \langle \SG_{1},s^0_1,\rightarrow_1\rangle$ and
  $S_2 = \langle \SG_{2},s^0_1,\,\allowbreak \rightarrow_2\rangle$ are
  two TTS with respective alphabets $\Sigma_1$ and $\Sigma_2$. The
  product of $S_1$ by $S_2$ is the TTS
  $S_1 \| S_2 = \langle (\SG_1 \times \SG_2), (s^0_1, s^0_2), \trans{}
  \rangle$ such that $\trans{}$ is the smallest relation obeying the
  following rules:
  \begin{mathpar}
    \inferrule{s_1 \trans{\alpha}_1 s'_1\\ \alpha \in (\Sigma_1 \setminus
      \Sigma_2) \cup \{ \epsilon\}} {(s_1, s_2) \trans{\alpha} (s_1',
      s_2)}

    \inferrule{s_2 \trans{\alpha}_2 s'_2\\ \alpha \in (\Sigma_2 \setminus
      \Sigma_1) \cup \{ \epsilon\}} {(s_1, s_2) \trans{\alpha} (s_1,
      s'_2)}

    \inferrule{s_1 \trans{\alpha}_1 s'_1\\ s_2 \trans{\alpha}_2 s'_2\\
      \alpha \neq \epsilon} {(s_1, s_2) \trans{\alpha} (s'_1, s'_2)}
  \end{mathpar}
\end{definition}

\noindent\textbf{Executions, Traces and Equivalences.}
An \emph{execution} of a net $N$ is a sequence in its semantics,
$\interp{N}$, that starts from the initial state. It is a time-event
word over the alphabet containing both labels (in
$\Sigma \cup \{\epsilon\}$) and delays.  Continuous transitions can
always be grouped together, meaning that when
$(m, \tI) \trans{\theta} (m, \tI')$ and
$(m, \tI') \trans{\theta'} (m, \tI'')$ then necessarily
$(m, \tI) \trans{\theta + \theta'} (m, \tI'')$ (and the firing domain
$\tI'$ is uniquely defined from $\tI$ and $\theta$). Based on this
observation, we can always consider executions of the form
$\sigma \eqdef \theta_0\, a_0\ \theta_1\, a_1\, \dots$ where each
discrete transition is preceded by a single time delay.
By contrast, a \emph{trace} is the untimed word obtained from an
execution when we keep only the discrete actions. Then the language of
a \TPN is the set
of all its (finite) traces.

By definition, the language of a \TPN is prefix-closed; and it is
regular when the net is bounded~\cite{berthomieu1983enumerative}. It
is also the case~\cite{refIPTPN} that the ``intersection'' of two nets
$N_1$ and $N_2$---the traces obtained from (pairs of) executions
common to the two nets---are exactly the traces in the TTS product
$\interp{N_1} \mathbin{\|} \interp{N_2}$. Our goal, in the next
section, is to define a product operation, $N_1 \times N_2$, that is a
\emph{congruence}, meaning that $\interp{N_1 \times N_2}$ should be
equivalent to $\interp{N_1} \mathbin{\|} \interp{N_2}$.

Language equivalence would be too coarse in this context. In this
paper, we will instead prefer (a weak version of) timed bisimulation,
which rely on a weak version of the transition relation
$s \wtrans{\alpha} s'$ (with $\alpha$ an action in
$\Sigma \cup \{\epsilon\} \cup \pRat$ and $\theta$ a delay in $\pRat$)
defined from the following set of rules:
\begin{mathpar}
  \inferrule{{\ }}
  {s \dwtrans{\epsilon} s}

  \inferrule{s \wtrans{\epsilon} s'\\ s' \trans{\alpha} s''\\
  s'' \wtrans{\epsilon} s'''}
  {s \dwtrans{\alpha} s'''}

  \inferrule{s \wtrans{\theta} s'\\
    s' \wtrans{\theta'} s''} {s \dwtrans{\theta + \theta'} s''}
\end{mathpar}

\begin{definition}[Behavioural Equivalence]\label{def:bisim}
  Assume $G_1 = \langle \SG_{1},s^0_1,\rightarrow_1\rangle$ and
  $G_2 = \langle \SG_{2},s^0_2,\rightarrow_2\rangle$ are two TTS. A
  binary relation $\arel$ over $\SG_{1} \times \SG_{2}$ is a weak
  timed bisimulation if and only if $s^0_1 \arel s^0_2$ and for all
  actions $\alpha$ and pair of states $(s_1, s_2) \in \,\arel$ we
  have: (1) if $s_1 \wtrans{\alpha} s'_1$ then there exists $s'_2$
  such that $s_2 \wtrans{\alpha} s'_2$ and $s'_1 \arel s'_2$ ; and
  conversely (2) if $s_2 \wtrans{\alpha} s'_2$ then there exists
  $s'_1$ such that $s_1 \wtrans{\alpha} s'_1$ and $s'_1 \arel
  s'_2$. In this case we say that $G_1$ and $G_2$ are timed bisimilar,
  denoted $G_1 \approx G_2$, and we use $\approx$ for the union of all
  timed bisimulations $\arel$.
\end{definition}

Timed bisimulation is preserved by product~\cite{refIPTPN}, meaning
that for all TTS $G, G_1$ and $G_2$ we have $G_1 \approx G_2$ implies
$(G \| G_1) \approx (G \| G_2)$. In the following we say that two nets
are bisimilar, denoted $N_1 \approx N_2$, when
$\interp{N_1} \approx \interp{N_2}$.\\

\begin{figure}[t]
\begin{center}
  \subfloat[$N_1$]{\label{fig:R}\includegraphics[scale=0.55]{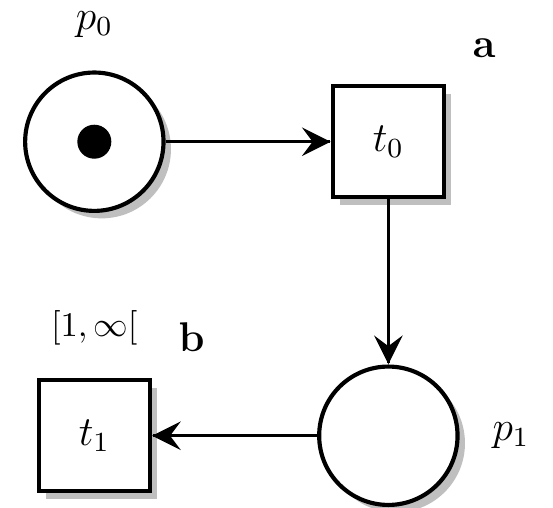}}
\quad
\vrule
\quad
\subfloat[$N_2$]{\label{fig:L}\includegraphics[scale=0.55]{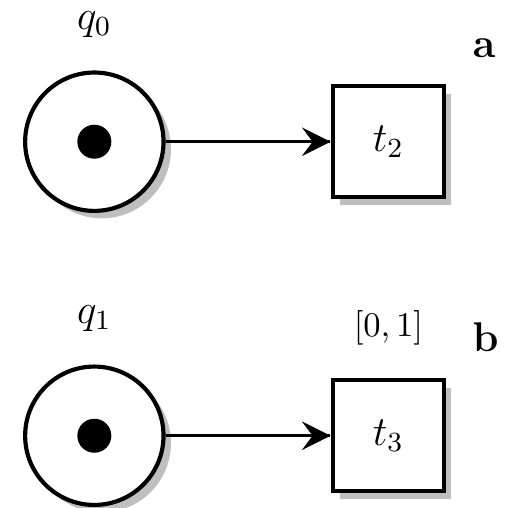}}
\quad
\vrule
\quad
\subfloat[``untimed'' product]{\label{fig:RL}\includegraphics[scale=0.55]{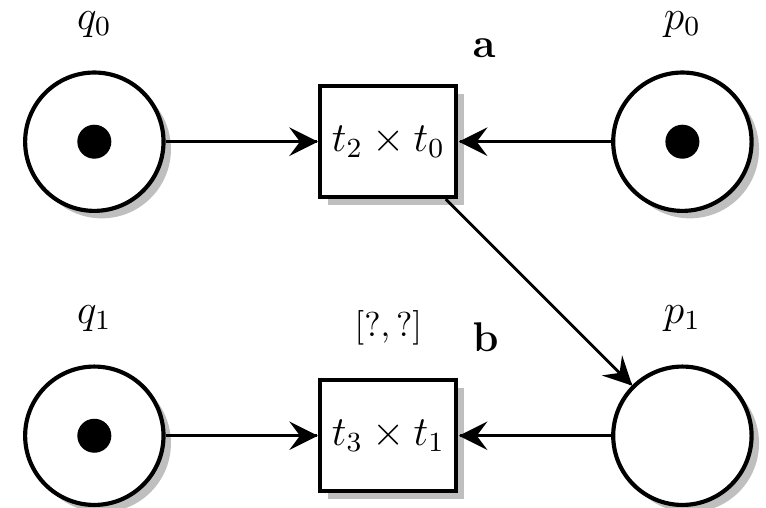}}
\end{center}
\caption{Two examples of \TPN and their (failed, untimed) product.}
\label{Figure-exeTPN}
\end{figure}

\noindent\textbf{Example.}
We give two examples of \TPN with alphabet $\{a,b\}$ in
Fig.~\ref{Figure-exeTPN}.  Executions for the net $N_1$ (left) include
time-event words of the form $\theta_0\, a\, \theta_1\, b$ (and their
prefix) provided that $\theta_1 \geq 1$. Executions for the net $N_2$ 
(middle) include time-event words of the form $\theta_2\, a\, \theta_3\, b$ and
$\theta_3\, b\, \theta_2\, a$ (and their prefix) provided that $\theta_3 \leq 1$.
If we consider executions that are in the product of both nets, we find all 
executions of the form $\theta_0\, a$, with the constraint $\theta_0 \leq 1$. We also
have one execution of the form $\theta_0\, a\, \theta_1\, b$ provided
that $\theta_0 + \theta_1 \leq 1$ and $\theta_1 \geq 1$. This
corresponds to the case where event $a$ fires exactly at date $0$; any
other case eventually leading to a \emph{time deadlock} (a situation
where time cannot progress). In the same figure (right), we display
the ``untimed'' synchronous product $N_1 \| N_2$. It is clear that
there are no possible choice of time constraint for transition
$t_3 \times t_1$ that could lead to a net bisimilar to
$\interp{N_1} \mathbin{\|} \interp{N_2}$. This is a simple example of
the ``non-composability'' of Time Petri nets.


\section{Product TPN and their Semantics}
\label{sec:2TPN}

We propose an extension of \TPN with a \emph{synchronous product}
operation between \TPN, $\times$, in the style of Arnold-Nivat
synchronization of processes~\cite{arnold_nivats_2002}. Our goal is to
obtain a congruent composition operator, in the sense that
$\interp{N_1 \times N_2} \approx \interp{N_1} \mathbin{\|}
\interp{N_2}$.
A \emph{product \TPN}, or \TTPN, is a \TPN
$\langle {P},{T},{\Pre},{\Post},\SIF \rangle$ augmented with two
projections, $\proj{1}$ and $\proj{2}$, such that the following
properties hold:
\begin{itemize}
\item there are two sets $\proj{1}P$ and $\proj{2}P$ that partition the set of
  places $P$.
\item there are two sets $\proj{1}T$ and $\proj{2}T$ that partition the set of
  transitions $T$.
\item all the pre- and post-conditions of a transition in $\proj{i}T$
  are places in $\proj{i}P$: if $t \in \proj{i}T$ and $\Pre(t)(p) > 0$
  or $\Post(t)(p) > 0$ then $p \in \proj{i}P$.
\end{itemize}
Basically, this means that a \TTPN $N$ is the superposition of two
distinct, non-interconnected components, that we call $\proj{1} N$ and
$\proj{2} N$ for short.

\begin{definition}[Product of \TPN]\label{def:ptpnprod}
  The product $N_1 \times N_2$ of two disjoint \TPN $N_1$ and $N_2$
  (such that $P_1 \cap P_2 = T_1 \cap T_2 = \emptyset$) is the \TTPN
  obtained from the juxtaposition, preserving labels, of $N_1$ and
  $N_2$ with the two trivial projections $\proj{i}P = P_i$ and
  $\proj{i}T = T_i$ for all $i \in 1..2$.
\end{definition}

With our notations, a \TTPN $N$ is equivalent to the composition
$(\proj{1} N) \times (\proj{2} N)$.
In the following, we use the notation $\proj{i} m$ to denote the
restriction of a marking $m$ to the places in $\proj{i}P$ and
similarly with $\proj{i}\varphi$ and the transitions in
$\proj{i}T$. By convenience, $\proj{i}(m,\tI)$ denotes the state
$(\proj{i}m,\proj{i}\tI)$ and we use $\proj{i}\Sigma$ for the alphabet
of net $\proj{i} N$.

To ease the presentation, we limit the composition to only two
components (instead of a sequence) and we do not define the equivalent
of ``synchronization vectors''. As a result, we do not define the
product over \TTPN. This could be added, at the cost of more
burdensome notations, but it is not needed in our applications
(Sect.~\ref{section:Twin}). This is also why we have the same
limitations in our implementation~\cite{twina2019}.

Labels are not necessarily partitioned, so the same label can be
shared between the two components of a product. We denote
$\Sigma_{1,2}$ the set $(\proj{1}\Sigma \cap \proj{2}\Sigma)$ of
labels occurring on ``both sides'' of a \TTPN. We should also need the
notation $\Sigma_{1}$ for the set
$(\proj{1}\Sigma \setminus \proj{2}\Sigma) \cup \{ \epsilon \}$ of
labels that can occur in $\proj{1}$ concurrently with $\proj{2}$ (and
similarly for $\Sigma_{2}$).  The semantics for \TTPN relies largely
on the semantics of \TPN but makes a particular use of labels.
\begin{definition}\label{def:ptpnstate}
  The semantics of a {\TTPN}
  $\langle {P},{T},{\Pre},{\Post},m_0,\SIF \rangle$, with projections
  $\proj{1}$ and $\proj{2}$, is the TTS
  $\interp{N}_\times = \langle \SG,s_0,\ptrans{}\rangle$ such that
  $s_0 = (m_0, \tI[0])$ is the same initial state than in the \TPN
  semantics $\interp{N}$, and $\ptrans{}$ is the transition relation
  with actions in ${\Sigma \cup \{\epsilon\} \cup \pRat}$ such that:
  \begin{mathpar}
    {\mprset{vskip=0.5ex}
      \inferrule*[]{\alpha \in \pRat \\\\ s \trans{\alpha} s' \in
        \interp{N}} {s \dptrans{\alpha} s'}

      \inferrule*[]{t \in T \quad
        \Lab(t)=\alpha \notin \Sigma_{1,2}\\\\ s \trans{\alpha} s' \in
        \interp{N}} {s \dptrans{\alpha} s'}

      \inferrule*[]{a = \Lab(t_1) = \Lab(t_2) \quad
        t_i \in \proj{i}T\\\\ \proj{i} s \trans{a} \proj{i} s' \in
        \interp{\proj{i}N}\quad i \in 1..2} {s \dptrans{a} s'}}
  \end{mathpar}
\end{definition}

The only new case is for pairs of transitions, $t_1$ and $t_2$ , from
different components but with the same label:
$\Lab(t_1)=\Lab(t_2)=a$. This is the equivalent of a
synchronization. Indeed the premisses entail that both $t_1$ and $t_2$
can fire immediately, and the effect is to fire both of them
simultaneously. As a side effect, our choice of semantics entails that
a transition on a ``shared label'' (in $\Sigma_{1,2}$) is blocked
until we find a matching transition, with the same label, on the
opposite component. This may introduce a new kind of time deadlock
that has no direct equivalent in a \TPN: when a transition has to fire
urgently (hence time cannot progress) while there are no matching
transition that is time-enabled.

It is the case that the reachable states, in $\interp{N}_\times$, are a
subset of the states in $\interp{N}$. This is because we may forbid a
synchronization on a shared label, but never create new opportunities
to fire a transition. We also have a more precise result concerning
the semantics of a \TTPN and the product of its components.
\begin{theorem}
  The TTS $\interp{N}_\times$ is isomorph to the product
  $\interp{\proj{1} N} \mathbin{\|} \interp{\proj{2} N}$.
\end{theorem}
\begin{proof}
  By induction on the shortest path from the initial state, $s_0$, to
  a reachable state $s$ in $\interp{N}_\times$ and then a case
  analysis on the possible transitions from $s$.\qed
\end{proof}

\section{Construction of the State Class Graph for \TTPN}
\label{sec:state-class-abstr}

We give a brief overview of the LSCG construction for a \TTPN
$N = \langle {P},{T},{\Pre},\allowbreak {\Post},m_0,\SIF \rangle$. In
the following, we use the notation $\alpha^s_t$ and $\beta^s_t$ for
the left and right endpoints of interval $\SIF(t)$. For the sake of
simplicity, we only consider inequalities that are non-strict (our
definitions can be extended to the more general case) and assume that
$\beta - \alpha = \infty$ when $\beta$ is infinite.

A \emph{state class} $C$ is a pair $(m, D)$, where $m$ is a marking
and $D$ is a \emph{domain}; a (finite) system of linear inequalities
on the firing dates of transitions enabled at $m$. We will use
variable $x_i$ in $D$ to represent the possible firing time of
transition $t_i$.
In the Linear SCG construction~\cite{BD91,berthomieu1983enumerative},
we build an inductive set of classes $C_\sigma$, where
$\sigma \in T^*$ is a sequence of discrete transitions firable from
the initial state. Intuitively, the class $C_\sigma = (m, D)$ collects
all the states reachable from the initial state by firing schedules of
support sequence $\sigma$. For example, the initial class
$C_{\epsilon}$ is $(m_{0}, D_0)$ where $D_0$ is the domain defined by
the static time constraints in $\tI[0]$, that is:
$\alpha^s_i \le \sI{i} \le \beta^s_i$ for all $t_i$ in
${\cal E}(m_0)$.

The efficiency of the SCG construction relies on several factors: (1)
First, we can restrict to domains $D$ that are \emph{difference
  systems}, that is a sequence of constraints of the form
$\alpha_i \leq x_i \leq \beta_i$ and $x_i - x_j \leq \gamma_{i,j}$,
where each variable in $(x_i)_{t_i \in {\cal E}(m_0)}$ corresponds to an enabled
transition (and $i \neq j$).
(2) Next, we can always put domains in \emph{closure form}, meaning
that each bounds $\alpha, \beta$ and $\gamma$ are the tightest
preserving the solution set of $D$. Hence we can encode $D$ using a
simple vector of values. This data structure, called \emph{Difference
  Bound Matrix} (DBM), is unique to all the domains that have equal
solution set. Hence testing class equivalence is decidable and
efficient.
(3) Finally, if $C_\sigma=(m,D)$ is defined and $t$ is enabled at $m$,
we can incrementally compute the coefficients of the DBM $D'$, the
domain obtained after firing $t$ from $C_\sigma$, from the
coefficients of $D$.

We only consider the new case where we simultaneously fire a pair of
transitions $(t_i, t_j)$ from a class $(m, D)$. We assume that the
resulting marking is $m'$.
First, we need to check that both transitions can eventually
fire. This is the case only if the condition $\gamma_{t,k} \geq 0$ is
true for all $t \in \{i, j\}$ and $k$ enabled at $m$ (with
$k \neq t$). In this case, the resulting domain $D'$ can be obtained
by following a short number of steps, namely:
\begin{enumerate}
\item add the constraints $x_i = x_j$ and $x_i \leq x_k$ to $D$, for
  all $k \notin \{i, j\}$ (since $t_i, t_j$ must fire at the same date
  and before any other enabled transition);
\item introduce new variables $x'_k$ for all transitions enabled in
  $m'$, that will become the variables in $D'$, and add the constraint
  $x'_k = x_k - x_i$ if $t_k$ is persistent or
  $\alpha^s_k \leq x'_k \leq \beta^s_k$ if $t_k$ is newly enabled;
\item eliminate all the variables from $D$ relative to transitions in
  conflict with $t_i, t_j$ and put the resulting system in normal
  form.
\end{enumerate}

Except for step 1 above, with the constraint that $x_i = x_j$, this is
exactly the procedure described in~\cite{BD91} for plain \TPN. When
both transitions $(t_i, t_j)$ can fire, it is possible to completely
eliminate all occurrences of the ``unprimed'' variables $x_k$ in $D'$
and the result is a DBM. Which is exactly what is needed in our case.

We can draw two useful observations from this result. First, we can
follow the same procedure with any number of equality constraints, and
still wind up with a DBM. Therefore it would be possible to fire more
than two transitions simultaneously. Second, we have an indirect proof
that forcing the synchronization of transitions is strictly less
constraining than using priorities (because it is not possible to use
the LSCG construction with priorities), something that was not obvious
initially.

\section{Expressiveness Results}
\label{sec:expressiveness}

It is not obvious that \TTPN add any expressive power compared to
\TPN.  On the one hand, the semantics of a \TTPN $N$ is quite close to
the semantics of its components. In particular,
$\interp{N}_\times = \interp{N}$ when there are no shared labels
($\Sigma_{1,2} = \emptyset$). Moreover, in a \TTPN like in a \TPN, it
is not possible to lose the ability of firing a transition just by
waiting long enough; a behaviour that distinguishes \TPN from TA, or
from \TPN with priorities for instance. On the other hand, \TTPN
introduces new kind of time deadlocks which are affected by time
delays (see our example at the end of Sect.~\ref{sec:TPN}).
Next, we prove that the two models are equally expressive (up-to
$\approx$) when all timing constraints are either infinite or closed
on the right (in which case we say the net is \emph{right-closed}).
\begin{theorem}\label{th:expressiveness-1}
  Given a safe, right-closed \TTPN $N$, we can build a safe,
  composable \TPN $N'$, whose size is linear with respect to $N$, such
  that $\interp{N}_\times \approx \interp{N'}$.
\end{theorem}

For the sake of brevity, we only sketch the proof. We rely on two
auxiliary properties and on an encoding from \TPN into
\emph{composable} net; meaning an equivalent net where all timing
constraints have been ``moved'' to silent transitions. We find such
result in~\cite[Def.~9]{refIPTPN}, which provides a construction to
build a composable net $\TT(N)$ from every safe and right-closed \TPN
$N$. Our restrictions on $N$ in Th.~\ref{th:expressiveness-1} come
from this construction, as is our result on the size of $N'$.

Our first auxiliary property, (L1), compare the product of composable
\TPN with their synchronous product, namely:
if $N_1$ and $N_2$ are composable \TPN then
$\interp{N_1 \times N_2}_\times \approx \interp{N_1 \, \| N_2}$.
Property (L1) derives directly from the construction of the product
$N_1 \, \| N_2$ of composable \TPN. Indeed, with composable nets, the
fusion of transitions sharing a common label are unaffected by
continuous transitions. Hence they have the same behaviour in
$N_1 \times N_2$ than in $N_1 \| N_2$. (And this is the only place
where the semantics of the two nets may diverge.)

Next, we use an equivalent of the congruence property for \TTPN, (L2):
given two pairs of \TPN $(N_1, N_2)$ and $(M_1, M_2)$ such
$N_1 \approx N_2$ and $M_1 \approx M_2$ we have that
$\interp{N_1 \times M_1}_\times \approx \interp{N_2 \times
  M_2}_\times$. Property (L2) can be proved by defining a ``candidate
relation'', $\arel$, which contains the pair $(s_0, s'_0)$ of initial
states of ${N_1 \times M_1}$ and ${N_2 \times M_2}$; then proving that
$\arel$ is a weak timed bisimulation. A suitable choice for $\arel$
is to take the smallest relation such that
$(s_1 \uplus s'_1) \arel (s_2 \uplus s'_2)$ whenever $s_1 \approx s_2$
and $s'_1 \approx s'_2$. Then the proof follows by simple case
analysis.

Finally, we use construction $\TT$ (above) to build composable \TPN
from the nets $\proj{1}N$ and $\proj{2}N$ and to define
$N' \eqdef \TT(\proj{1}N) \| \TT(\proj{2}N)$.
By property of $\TT$ we have $\proj{i}N \approx \TT(\proj{i}N)$ for
all $i \in 1..2$. Hence by (L2) and (L1) we have
$\interp{N}_\times = \interp{\proj{1}N \times \proj{2}N}_\times
\approx \interp{\TT(\proj{1}N) \times \TT(\proj{2}N)}_\times \approx
\interp{\TT(\proj{1}N) \mathbin{\|} \TT(\proj{2}N)}$. The property
follows by transitivity of $\approx$.

Our proof gives a constructive method to build a net $N'$ with (at
most) four extra transitions and places, compared to $N$, for each
non-trivial labeled transition. We can use the SCG of $N'$ to compute
the language of $N$ (and to compute the intersection of two nets when
we choose $N = N_1 \times N_2$). Unfortunately this approach does not
scale well. For example, the composition of the two nets given in
Fig.~\ref{Figure-exeTPN} has 16 classes with this method instead of
only $3$ with our approach (and the intermediary \TPN has 11 places
and 7 transitions). Likewise, for the simple example in
Fig.~\ref{Figure-image3} we have a net with 25 places, 211 transitions
and $1\,389$ classes instead of simply $3$ classes with \TTPN.

Another limitation of this approach are the restrictions imposed on the
timing constraints of $N$. Indeed, to the best of our knowledge, there
are no equivalent of construction $\TT$ in the case of ``right-open''
transitions.\\

\noindent\textbf{Composable Time Petri nets using IPTPN.}
Berthomieu et al.~\cite{refIPTPN} define an extension of \TPN with
``inhibition and permission'' that provides another method for
building composable nets. With this extension, it is always possible
to build a composable IPTPN from a \TPN. For example,
Fig.~\ref{fig:iptpn} displays the IPTPN corresponding to the
``product'' of the two nets in Fig.~\ref{Figure-exeTPN}. In this
construction, we create a silent, extra-transition $tc_i$ for every
non-trivial observable transition $t_i$. These transitions cannot fire
(they self-inhibit themselves with an
\tikz{\node[circle,inner sep=0pt](p0) at (0,0){};
  \node[circle,draw=orange,inner sep=0pt,minimum size=4pt,thick](p1)
  at (4ex, 0){}; \draw[orange, thick](p0) -- (p1);}\/ arc)
but ``record the timing constraints'' of the transition they are
associated with. Then a permission arc (\/%
\tikz{\node[circle,inner sep=0pt](p0) at (0,0){};
  \node[circle,draw=orange,fill=orange, inner sep=0pt,minimum size=4pt,thick](p1) at
  (4ex, 0){}; \draw[orange, thick](p0) -- (p1);}\/)
is used to transfer these constraints on the (product of) labeled
transitions.

Tina provides a SCG construction for IPTPN but, like with the addition
of priorities, it is necessary to use the {strong} construction in
this case. We use the encoding into IPTPN we just sketched above in
our experiments.

\begin{figure}[t]
\centering
\begin{minipage}{.45\textwidth}
  \centering
  \includegraphics[scale=0.55]{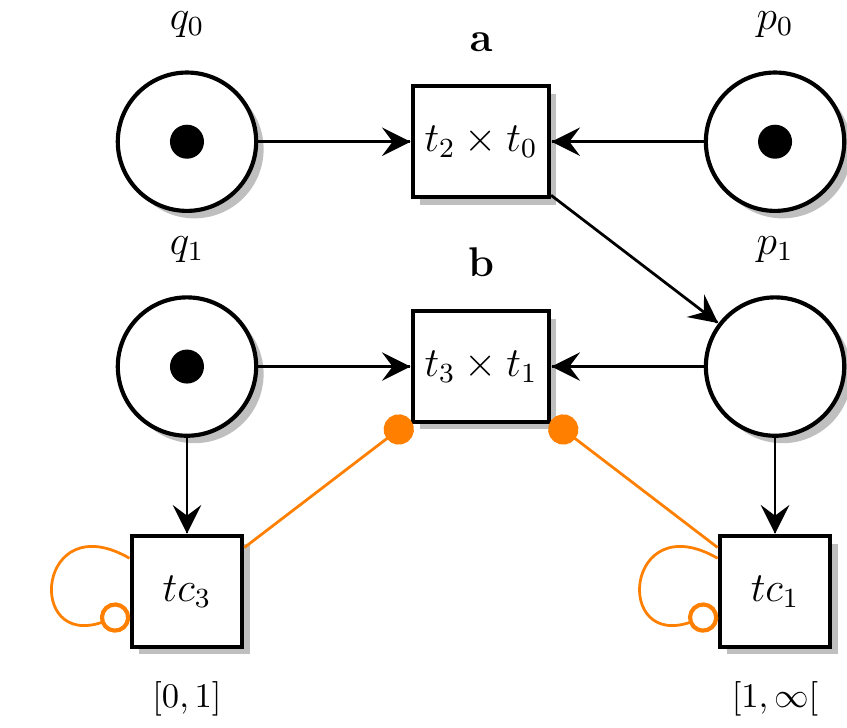}
  \captionof{figure}{Example of IPTPN}
  \label{fig:iptpn}
\end{minipage}%
\hfill\vrule\hfill%
\begin{minipage}{.45\textwidth}
  \centering
  \includegraphics[scale=0.55]{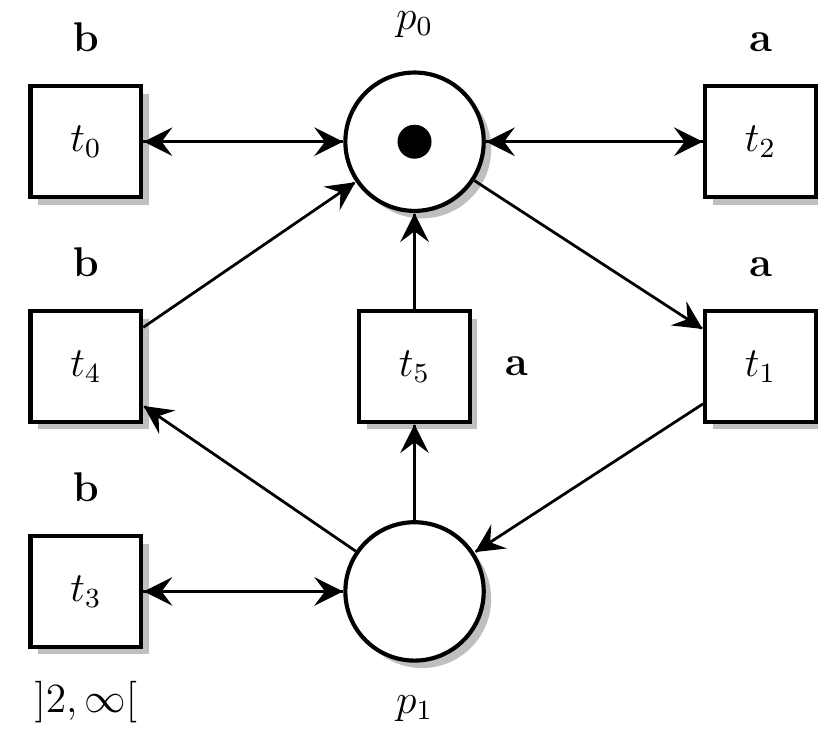}
  \captionof{figure}{\TPN for the delay property}
  \label{fig:obs}
\end{minipage}
\end{figure}

\section{Experimental Results and Possible Applications}
\label{section:Twin}
\label{sec:experimental-results}

We have implemented the state class construction for \TTPN in a tool
called Twina~\cite{twina2019} that can generate the LSCG of both
``plain'' and product \TPN. The tool and models mentioned here are
available online at \url{https://projects.laas.fr/twina/}, with
instructions on how to reproduce our results.\\

\noindent\textbf{Performances Compared with IPTPN}
We compare the results obtained with \TTPN and an encoding into IPTPN,
which appears to be the best alternative among the three methods
mentioned in Sect.~\ref{sec:introduction}.  By default, Twina uses
option {-W}, that computes the Linear SCG of a net. We also provide
option {-I} to compute the LSCG for the product of two nets using the
construction defined in Sect.~\ref{sec:state-class-abstr}.  We use the
same syntax for nets in Twina than in Tina~\cite{BRV04}. In
particular, our method can be used with nets that are not 1-safe and
without any restriction on the timing constraints (so we accept
right-open transitions). We also allow read- and inhibitor-arcs with
the same semantics than in Tina.
We compare the size of the LSCG
with the results obtained using IPTPN and Tina in
Table~\ref{tab:1}. The results are reported with the sizes of the SCG
in number of classes and edges; we also give the ratio of classes
saved between the LSCG and the SSCG. So a $100\%$ ratio means twice as
much states in the strong SCG.

We use different models for our benchmarks: \textit{jdeds} is an
example taken from~\cite{gougam2017diagnosability} extended with time;
\textit{train} is a modified version of the train controller example
in~\cite{Ber03b} with an additional transition that corresponds to a
fault in the gate;
\noindent\textit{plant} is
the model of a complex automated manufacturing system
from~\cite{wang2015diagnosis}; \textit{wodes}~is the WODES diagnosis benchmark of Giua (see
e.g.~\cite{cabasino2009discrete}) with added timed constraints.
For each model, we give the result of three experiments:
\textit{plain} where we compute the SCG of the net, alone;
\textit{twin} where we compute the intersection between the \TPN and a
copy of itself with some transitions removed; and \textit{obs} where
we compute the intersection of the net with a copy of the \TPN in
Fig.~\ref{fig:obs}. We explain the relevance of the last two
constructions just afterwards.

{\setlength{\tabcolsep}{4pt} 
  \begin{table}[t]
    \centering

    \pgfplotstableset{
      col sep=comma,
      columns/model/.style={
        string type,
        column type=|l,
        column name=\textsc{Model}},
      columns/exp/.style={
        string type,
        column type=l|,
        column name=\textsc{Exp.}},
      columns/statel/.style={
        numeric type,
        fixed,
        column type=|r|,
        column name=\textsc{States}},
      columns/transl/.style={
        numeric type,
        fixed,
        column type=r|,
        column name=\textsc{Trans.}},
      columns/states/.style={
        numeric type,
        fixed,
        column type=|r|,
        column name=\textsc{States}},
      columns/transs/.style={
        numeric type,
        fixed,
        column type=r,
        column name=\textsc{Trans.}},
      create on use/ratio/.style={
        create col/expr={(\thisrow{states}-\thisrow{statel})*100/\thisrow{statel}}
      },
      columns/ratio/.style={
        column type=||r|,
        precision=0,
        postproc cell content/.append code={
          \pgfkeysalso{@cell content/.add={}{\%}}%
        },
        fixed,
        fixed zerofill,
        column name=\textsc{Ratio}},
      every head row/.style={
        before row={
          \hline
          \multicolumn{2}{|c||}{} &
          \multicolumn{2}{c||}{Twina (LSCG)} &
          \multicolumn{2}{c||}{IPTPN (SSCG)} &
          \multicolumn{1}{c|}{}\\ 
          \cline{3-6}
        }, 
        after row={
          \hline
        }
      },
      every last row/.style={after row=\hline}, 
      every nth row={3}{before row=\hline},
    }

    \pgfplotstableread{data.csv}\loadedtable

    \pgfplotstabletypeset[
    columns={model,exp,statel,transl,states,transs,ratio},
    1000 sep={\,},
    ]\loadedtable\\[1em]
    \caption{Comparing the \TTPN and IPTPN methods}
    \label{tab:1}
  \end{table}
}

We see that, in some of our examples, there is a large difference
between the size of the LSCG and the size of the SSCG for the same
example. This was one of our main reason for developing a specific
tool. This is important since, on the extreme case, we can have a
quadratic blow-up in the number of classes when analysing a twin
product. (This is almost the case in example {jdeds}.) We also observe
that, on model {plant-twin}, the size of the intersection may be much
smaller than the size of one of the component alone; $1300$ classes
compared to $2$ million. This is to be expected, since the
intersection may have only one class. Nonetheless this emphasizes the
need to have methods that can build the intersection on the fly,
without computing a symbolic
representation for each component first.\\

\noindent\textbf{Diagnosability and the Twin Plant Method.}
One possible application of Twina---and our initial motivation for
this work---is to check \emph{fault
  diagnosability}~\cite{sampath1995diagnosability} in systems modelled
as \TPN~\cite{basile2017,Cabasino12}.  In this context, a system is
described as a \TPN with a distinguished unobservable event $f$ that
models a fault.  (Any transition labelled with $f$ is faulty.)  Fault
$f$ is diagnosable if it is always possible to detect when a faulty
transition has fired, in a finite amount of time, just by looking at
the observable flow of events~\cite{Tripakis02}. Under the assumptions
that the system does not generate Zeno executions, and that any
possible execution is not infinitely unobservable, one way to check
diagnosability is to look for infinite \emph{critical
  pairs}~\cite{cimatti_2003_refid132}. A critical pair consists of a
couple of infinite executions of the \TPN, one faulty the other one
not, that have equal timed observations. Then fault $f$ is diagnosable
if no such pair exists. By using Twina, we aim at checking
diagnosability by adapting the \emph{twin-plant}
method~\cite{refTwin:poly} to \TPN.  The idea is to make two copies of
the same system, one with the fault, $N_f$, and the other without it,
$N_o$, and to relabel all unobservable events to avoid collisions.
Then checking for the existence of an infinite critical pair amounts
to finding an infinite execution with $f$ in the product
$N_o \times N_f$.
\begin{figure}[t]
   \begin{center}
\raisebox{-0.5\height}{\subfloat{\includegraphics[scale=0.5]{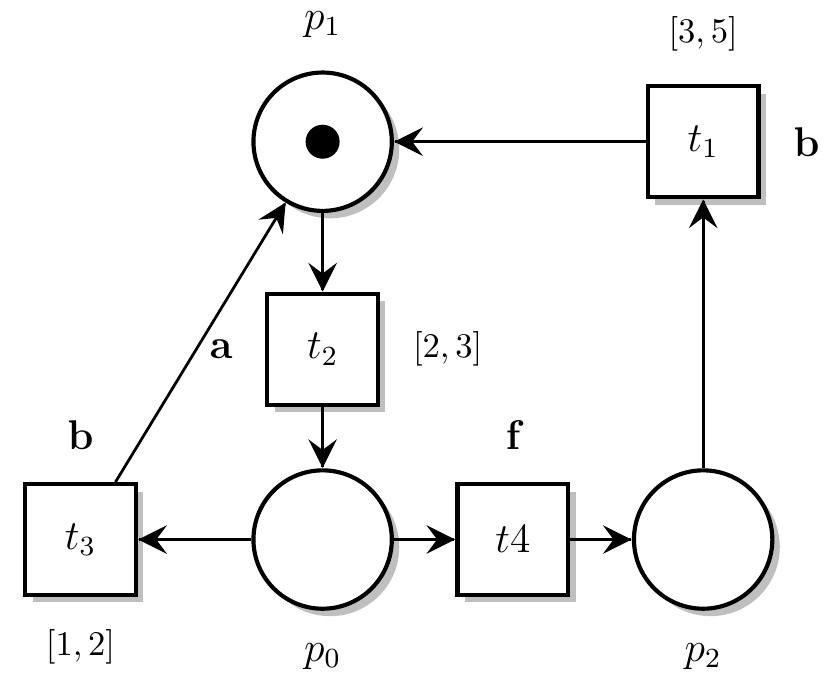}}}
     \qquad\scalebox{2}{$\times$}\qquad
\raisebox{-0.5\height}{\subfloat{\includegraphics[scale=0.5]{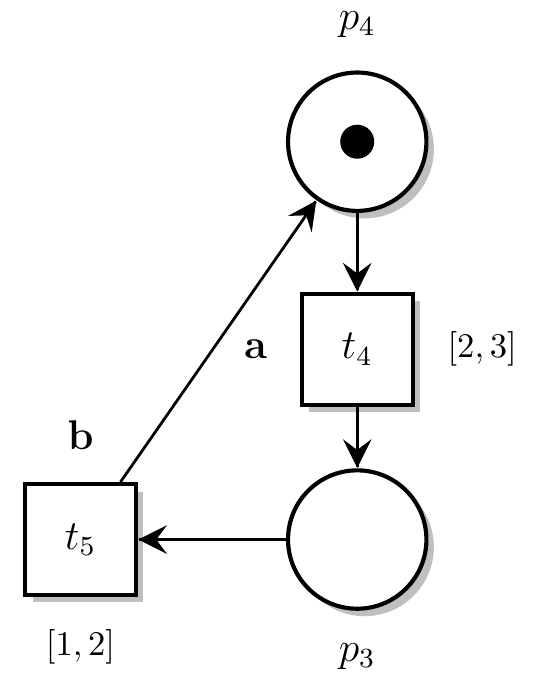}}}
   \end{center}
   \caption{Product of a \TPN $N_f$ (left) and its ``twin'' $N_o$
     (right) for the fault ${f}$.}
   \label{Figure-image3}
\end{figure}

We give an example of this construct in Fig.~\ref{Figure-image3} where
$a$ and $b$ are both observable.
In system $N_f$, fault $f$ is not diagnosable if we do not consider
time, as we always observe $b$ after an $a$ in both faulty and
non-faulty executions. Now considering the observation of time, $f$ is
diagnosable as the date of $b$ is always discriminant. In the
intersection of $N_f$ and $N_o$, every execution where transition
$t_4$ fires leads to a time deadlock. Indeed, in this case, we must
wait at least $3$ to fire transitions $t_1$ and at most $2$ to fire
$t_5$ (and both have label $b$).

The twin-plant construction is quite useful and we provide an option
to directly build a twin \TPN in our tool (option {-twin}). This is
the construction we use in our experiments for Table~\ref{tab:1}. In
this case, we can generate a LTS for the twin plant and check that
every fault eventually leads to a deadlock in the product, meaning that 
the system is diagnosable. For instance using a LTL
model-checker and a property such as
$(\lozenge f) \Rightarrow (\lozenge
\mathrm{dead})$~\cite{gougam2017diagnosability}. We also provide a
dedicated algorithm (option {-diag}) to check this property
on-the-fly. When the system is not diagnosable, it allows us to find a
counter-example before exploring the whole behaviour of the
twin-plant.\\

\noindent\textbf{Observer-based verification.} Another
application of our product construction is model checking \TPN, in
much the same way some ``observer-based'' verification techniques rely
on the product of a system with an
observer~\cite{DalzilioS:IJCCBS2014verifpatterns,toussaint1997time}. The
idea is to express a property as the language of an observer, $O$,
then check the property on the system $N$ by looking at the behaviour
of $N \times O$. A major advantage of this approach is that there is
no risk to disrupt the system under observation, which is not always
easy to prove with other methods.

We give an example of observer in Fig.~\ref{fig:obs}. In this net,
sequences of events $a$ and $b$ may occur in any order and at any
date. On the other hand, the only way to fire $t_3$ is to ``find'' two
successive occurrences of $a$ and $b$ with a delay (strictly) bigger
than $2$. Hence we can check if such behaviour is possible in a
system, $N$, by checking whether $t_3$ can fire in $N \times O$. This
is the problem we consider in the \textit{obs} experiments of
Table~\ref{tab:1}. We only consider one small example
here. Nonetheless, the same approach could be used to check more
complex timed properties. This will be the subject of future works.

\section{Conclusion}

We propose an extension of \TPN with a product operation in the style
of Arnold-Nivat. The semantics of our extension is quite
straightforward. What is more surprising is that it is possible to
adapt the LSCG construction to this case---which means that we do not
need the equivalent of clocks or priorities---and that this extension
does not add any expressive power. This is a rather promising result,
complexity-wise, since it means that we can hope to adapt the same
optimization techniques than with ``plain'' \TPN, such as specific
symmetry reduction techniques for
instance~\cite{DalzilioS:scp2016symmetries}.

We have several opportunities for extending our work. Obviously we can
easily extend our product to a sequence of nets and add a notion of
``synchronization vectors''. This could lead to a more compositional
framework for \TPN, in the style of the BIP
language~\cite{basu:hal-00375298}. Another promising application of
our approach would be to extend classical results from the theory of
supervisory control to the context of \TPN. We already mentioned a
possible application for diagnosability (which was the initial
motivation for our work). A next step could be to study the
``quotient'' of two \TPN language---the dual of the product---which
can be used to reason about the controlability of a system and that is
at the basis of many
compositional verification methods, such as Assume-Guarantee for example.\\

\noindent\textbf{Acknowledgments.} The authors are grateful to Thomas
Hujsa and Pierre-Emmanuel Hladik for their valuable comments. We also
want to thank Bernard Berthomieu, without whom none of this would have
been possible; our work is a tribute to the versatility and the
enduring qualities of the state class construction that he pioneered
more than 30 years ago.

\bibliographystyle{splncs04}
\bibliography{bibli}

\end{document}